\documentclass[11pt]{article}
\usepackage{fullpage}
\usepackage{float}
\usepackage{epsfig}
\usepackage{amsmath}
\usepackage{amssymb}
\usepackage{amstext}
\usepackage{xspace}
\usepackage{latexsym}
\usepackage{verbatim}
\usepackage{multirow}
\usepackage{ifthen}
\usepackage{url}
\usepackage{hyperref}

\usepackage[ruled,vlined]{algorithm2e}
\renewcommand{\algorithmcfname}{ALGORITHM}

\DeclareMathAlphabet{\mathitbf}{OML}{cmm}{b}{it}

\usepackage{times}

\newtheorem{theorem}{Theorem}[section]

\newtheorem{claim}{Claim}
\newtheorem{lemma}[theorem]{Lemma}

\newtheorem{corollary}[theorem]{Corollary}

\newcommand{\qed}{\mbox{\ \ \ }\rule{6pt}{7pt} \bigskip}

\renewcommand{\comment}[1]{}
\newenvironment{proof}{\noindent{\em Proof:}}{\hfill\qed}

\makeatletter
\newenvironment{oneshot}[1]{\@begintheorem{#1}{\unskip}}{\@endtheorem}
\makeatother

\newcommand{\dist}{F}

\newcommand{\distj}[1][j]{{\dist_{#1}}}

\newcommand{\expect}{\mathop{\operatorname{\bf E}}}
\newcommand{\pr}{\mathop{\operatorname{\bf Pr}}}
\renewcommand{\Pr}{\pr}
\newcommand{\ex}{\expect}

\newcommand{\VCGc}{\ensuremath{\mathrm{BO}_{\overload}}}

\newcommand{\OPT}{\ensuremath{\mathrm {OPT}}}
\newcommand{\OPTh}{\ensuremath{\mathrm{OPT}_{1/2}}}
\newcommand{\OPTr}[1][\partparam]{\ensuremath{\mathrm{OPT}_{#1}}}

\newcommand{\msp}{\text{Makespan}}
\newcommand{\tj}{\alpha_j}
\newcommand{\load}{\eta}

\newcommand{\rank}{r}
\newcommand{\T}{T}
\newcommand{\Tj}{T_j}
\newcommand{\Tji}[1][i]{\Tj\ithofn{#1}{m}}
\newcommand{\Bj}{{B_j}}

\newcommand{\dmax}[2]{#2\maxn{#1}}
\newcommand{\dmin}[2]{#2\minn{#1}}
\newcommand{\dsum}[2]{#2\sumn{#1}}

\newcommand{\Rj}{R_j}
\newcommand{\Lj}{L_j}
\newcommand{\Rhatj}{\tilde{R}_j}

\newcommand{\G}{\mathcal{G}}
\newcommand{\B}{\mathcal{B}}

\newcommand{\overload}{c}
\newcommand{\overloadvalue}{7}
\newcommand{\finalApprox}{200}
\newcommand{\finalApproxMHR}{800} 
\newcommand{\rvw}{W}
\newcommand{\reserve}{\beta}
\newcommand{\partparam}{\delta}

\newcommand{\ithofn}[2]{[#1{:}#2]}
\newcommand{\sumn}[1]{[\Sigma#1]}
\newcommand{\maxn}[1]{\ithofn{#1}{#1}}
\newcommand{\minn}[1]{\ithofn{1}{#1}}

\newcommand{\couplingApprox}{4}

\begin{document}
\title{Prior-Independent Mechanisms for Scheduling}
\author{Shuchi Chawla\thanks{Computer Sciences Dept., University of Wisconsin -
  Madison. \tt{shuchi@cs.wisc.edu}.}
 \and Jason D. Hartline\thanks{EECS, Northwestern
   University. \tt{hartline@eecs.northwestern.edu}.}
 \and David Malec\thanks{Computer Sciences Dept., University
 of Wisconsin - Madison. \tt{dmalec@cs.wisc.edu}.}
 \and Balasubramanian Sivan\thanks{Computer Sciences Dept., University
 of Wisconsin - Madison. \tt{balu2901@cs.wisc.edu}.}
}
\date{}
\maketitle{}
\begin{abstract}
We study the makespan minimization problem with unrelated selfish
machines under the assumption that job sizes are stochastic. We design
simple truthful mechanisms that under various distributional
assumptions provide constant and sublogarithmic approximations to
expected makespan. Our mechanisms are prior-independent in that they
do not rely on knowledge of the job size distributions.
Prior-independent approximation mechanisms have been previously
studied for the objective of revenue maximization~\cite{DRY10, DHKN11,
RTY12}. In contrast to our results, in prior-free settings no truthful
anonymous deterministic mechanism for the makespan objective can
provide a sublinear approximation~\cite{ADL09}.

\end{abstract}
%
%

\thispagestyle{empty}
\newpage

\section{Introduction}
We study the problem of scheduling jobs on machines to minimize
makespan in a strategic context. The makespan the longest it takes any
of the machines to complete the work assigned by the schedule. The
running time or size of a job on a machine is drawn from a fixed
distribution, and is a private input known to the machine but not to
the optimizer. The machines are unrelated in the sense that the
running time of a job on distinct machines may be distinct.  A
scheduling mechanism solicits job running times from the machines and
determines a schedule as well as compensation for each of the
machines. The machines are strategic and try to maximize the
compensation they receive minus the work they perform.  We are
interested in understanding and quantifying the loss in performance
due to the strategic incentives of the machines who may misreport the
job running times.


A primary concern in the theory of mechanism design is to understand
the compatibility of various objectives of the designer with the
incentives of the participants. As an example, maximizing social
welfare is incentive compatible; the Vickrey-Clarke-Groves (VCG)
mechanism obtains this socially optimal outcome in
equilibrium \cite{V61,C71,G73}.  For most other objectives,
however, the optimal solution ignoring incentives (a.k.a.\@ the {\em
first-best} solution) cannot be implemented in an incentive compatible
manner. This includes, for example, the objectives of revenue
maximization, welfare maximization with budgets, and makespan
minimization with unrelated machines.  For these objectives there is
no incentive compatible mechanism that is best on every input.  The
classical economic approach to mechanism design thus considers inputs
drawn from a distribution (a.k.a.\@ the {\em prior}) and looks for the
mechanism that maximizes the objective in expectation over the
distribution (a.k.a.\@ the {\em second-best} solution).

The second-best solution is generally complex and, by definition,
tailored to specific knowledge that the designer has on the
distribution over the private information (i.e., the input) of the
agents.  The non-pointwise optimality, complexity, and distributional
dependence of the second-best solution motivates a number of mechanism
design and analysis questions.
\begin{description}
\item[price of anarchy:] 
For any distribution over inputs, bound the gap between the first-best
(optimal without incentives) and second-best (optimal with incentives)
solutions (each in expectation over the input).

\item[computational tractability:] 
For any distribution over inputs, give a computationally tractable
implementation of the second-best solution, or if the problem is
intractable give a computationally tractable approximation mechanism.

\item[simplicity:]
For any distribution over inputs, give a simple, practical mechanism
that approximates the second-best solution.

\item[prior independence:] 
Give a single mechanism that, for all distributions over inputs,
approximates the second-best solution.
\end{description}
These questions are inter-related. As the second-best mechanism is
often complex, the price of anarchy can be bounded via a lower bound
on the second-best mechanism as given by a simple approximation
mechanism.  Similarly, to show that a mechanism is a good
approximation to second-best the upper bound given by the first-best
solution can be used.  Importantly though, if the first-best solution does
not permit good approximation mechanisms then a better bound on the
second-best solution should be sought.  Each of the questions above
can be further refined by consideration with respect to a large class of
priors (e.g.\@ identical distributions).

The prior-independence question gives a middle ground between
worst-case mechanism design and Bayesian mechanism design.  It
attempts to achieve the best of both worlds in the tradeoff between
informational efficiency and approximate optimality.  Its minimal
usage of information about the setting makes it robust.  A typical
side-effect of this robustness is simple and natural mechanisms;
indeed, our prior-independent mechanisms will be simple,
computationally tractable, and also enable a bound on the price of
anarchy.

The literature on prior-independent mechanism design has focused
primarily on the objective of revenue maximization.  Hartline and
Roughgarden \cite{HR09} show that with sufficient competition, the
welfare maximizing (VCG) mechanism also attains good revenue.  This
result enables the prior-independent approximation mechanism for
single-item auctions of Dhangwatnotai, Roughgarden, and
Yan~\cite{DRY10} and the multi-item approximation mechanisms of
Devanur et al.~\cite{DHKN11} and Roughgarden et al.~\cite{RTY12}.
Importantly, in single-item auctions the agents' private information
is single-dimensional whereas in multi-item auctions it is
multi-dimensional.  There are several interesting and challenging
directions in prior-independent mechanism design: (1) non-linear
objectives, (2) general multi-parameter preferences of agents, (3)
non-downwards-closed feasibility constraints, and (4) non-identically
distributed types of agents. Our work addresses the first three of
these four challenges.

%
%
We study the problem of scheduling jobs on machines where the runtime
of a job on a machine is that machine's private information.  The
prior over runtimes is a product distribution that is symmetric with
respect to the machines (but not necessarily symmetric with respect to
the jobs).  {\em Ex ante}, i.e., before the job sizes are
instantiated, the machines appear identical; {\em ex post}, i.e.,
after the job sizes are realized, the machines are distinct and job
runtimes are unrelated.  The makespan objective is to schedule the
jobs on machines so as to minimize the time at which the last machine
completes all of its assigned jobs.  Our goal is a prior-independent
approximation of the second-best solution for the makespan objective.




To gain intuition for the makespan objective, consider why the simple
and incentive compatible VCG mechanism fails to produce a good
solution in expectation.  The VCG mechanism for scheduling minimizes
the total work done by all of the machines and accordingly places
every job on its best machine. Note that because the machines are a
priori identical, this is an i.i.d.\@ uniformly random machine for
every job. Therefore, in expectation, every machine gets an equal
number of jobs. Furthermore, every job simultaneously has its smallest
size possible. However, the maximum load in terms of the number of
jobs per machine and so also the makespan can be quite large. The
distribution of jobs across machines is akin to the distribution of
balls into bins in the standard balls-in-bins experiment---when the
number of balls and bins is equal, the maximum loaded bin contains
$\Theta(\log n/\log\log n)$ balls with high probability even though
the average load is $1$.

%
%
Our designed mechanism must prevent the above balls-in-bins style
behavior.  Consider a variant of VCG that we call the {\em bounded
overload} mechanism.  The bounded overload mechanism minimizes the
total work with the additional feasibility constraint that the load
(i.e., number of jobs scheduled) of any machine is bounded to be at
most a $c$ factor more than the average load.  This mechanism is
``maximal in range'', i.e., it is simply the VCG mechanism with a
restricted space of feasible outcomes; it is therefore incentive
compatible.  Moreover, the bounded overload mechanism can be viewed as
belonging to a class of ``supply limiting'' mechanisms (cf.\@ the
prior-independent supply-limiting approximation mechanism
of \cite{RTY12} for multi-item revenue maximization).

While the bounded overload mechanism evens out the number of jobs per
machine, an individual job may end up having a running time far larger
than that on its best machine.  The crux of our analysis is to show
that this does not hurt the expected makespan of our schedule relative
to an ideal setting where every job assumes its minimum size.  Our
analysis of job sizes has two components. First we show that every job
with high probability gets assigned to one of its best
machines. Second, we show that the running time of a job on its $i$th
best machine can be related within a factor depending on $i$ to its
running time on its best machine. These components together imply that
the bounded overload mechanism simultaneously obtains a schedule that
is balanced in terms of the number of jobs per machine and where every
job has a small size (in comparison to the best possible for that
job). This is sufficient to imply a constant factor approximation to
expected makespan when the number of jobs is proportional to the
number of machines.


The second component of our analysis of job sizes in the bounded
overload mechanism entails relating different order statistics of
(arbitrary) i.i.d.~distributions, a property that may have broader
applications. In particular, letting $X\ithofn{k}{n}$ denote the $k$th
minimum out of $n$ independent draws from a distribution, we show that
for any $k$ and $n$, $X\ithofn{k}{n}$ is nearly stochastically dominated by
an exponential function of $k$ times $X\minn{n/2}$. In simple terms, the
minimum out of a certain number of draws cannot be arbitrarily smaller
than the $k$th minimum out of twice as many draws.

As an intermediary step in our analysis we bound the performance of our
approximation mechanism with respect to the first-best solution with
half the machines (recall, machines are a priori identical).  Within
the literature on prior-independent revenue maximization this approach
closely resembles the classical Bulow-Klemperer theorem \cite{BK96}.
For auctioning $k$ units of a single-item to $n$ agents (with values
drawn i.i.d.~from a ``nice'' distribution), the revenue from welfare
maximization exceeds the optimal revenue from $n-k$ agents.  In other
words, a simple prior-independent mechanism with extra competition
(namely, $k$ extra agents) is better than the prior-optimal mechanism
for expected revenue.  Our result is similar: when the number of jobs
is at most the number of machines and machines are a priori identical,
we present a prior-independent mechanism that is a constant
approximation to makespan with respect to the first-best (and
therefore also with respect to the second-best) solution with half as
many machines.  Unlike the Bulow-Klemperer theorem we place no
assumptions the distribution of jobs on machines besides symmetry with
respect to machines.  




To design scheduling mechanisms for the case where the number of jobs
is large relative to the number of machines we can potentially take
advantage of the law of large numbers.  If there are many more large
jobs (i.e., jobs for which the best of the machines' runtimes is
significant) then assigning jobs to machines to minimize total work
will produce a schedule where the maximum work on any machine is
concentrated around its expectation; moreover, the expected load of
any machine in the schedule that minimizes total work is at most the
expected load of any machine in the schedule that minimizes makespan.

On the other hand, if there are a moderate number, e.g., proportional
to the number of machines, of jobs with very large runtimes on all
machines, both the minimum work mechanism and the bounded overload
mechanism can fail to have good expected makespan.  For the bounded
overload mechanism, although the distribution of jobs across machines
is more-or-less even, the distribution of the few ``worst'' jobs that
contribute the most to the makespan may be highly uneven.  Indeed, for
a distribution where the expected number of large jobs is about the
same as the number of machines, the bounded overload mechanism
exhibits the same bad balls-in-bins behavior as the minimum work
mechanism.

The problem above is that the existence of many small, but relatively
easy to schedule jobs, prevents the bounded overload mechanism from
working.  To solve this problem we employ a two stage approach.  The
first stage acts as a sieve and schedules the small jobs to minimize
total work and while leaving the large jobs unscheduled.  Then in the
second stage the bounded overload mechanism is run on the unscheduled
jobs.   With the proper parameter tunings (i.e., job size threshold for
the sieve and partitioning of machines to the two stages) this
mechanism gives a schedule with approximately optimal expected
makespan.  We give two parameter tunings and analyses, one which gives
an $O(\sqrt{\log m})$ approximation and the other that gives an
$O((\log\log m)^2)$ approximation under a certain tail condition on
the distribution of job sizes (satisfied, for example, by all monotone
hazard rate distributions).

The proper tuning of the parameters of the mechanism require knowledge
of a single order statistic of the size distribution, namely the
expected size of a job on its best out of $k$ machines for an
appropriate value of $k$, to decide which jobs get scheduled in which
stage. This statistic can be easily estimated as the mechanism is
running by using the reports of a small fraction of the machines as
a ``market analysis.''  To keep our exposition and analysis simple, we
skip this detail and assume that the statistic is known.

\subsection*{Related work}

%
%
There is a large body of work on prior-free mechanism design for the
makespan objective.  This work does not assume a prior distribution,
instead it looks at worst-case approximation of the first-best
solution (i.e., the optimal makespan without incentive constraints).
The problem was introduced by Nisan and Ronen~\cite{NR99} who showed
that the minimum work (a.k.a.\@ VCG) mechanism gives an
$m$-approximation to makespan (where $m$ is the number of machines).
They gave a lower bound of two on the worst case approximation factor
of any dominant strategy mechanism for unrelated machine scheduling.
They conjectured that the best worst-case approximation is indeed
$\Theta(m)$.  Following this work, a series of papers presented better
lower bounds for deterministic as well as randomized
mechanisms~\cite{CKV07, CKK07, KV07, MS07}.  Ashlagi, Dobzinski and
Lavi~\cite{ADL09} recently proved a restricted version of the
Nisan-Ronen conjecture by showing that no {\em anonymous}
deterministic dominant-strategy incentive-compatible mechanism can
achieve a factor better than $m$.  This lower bound suggests that the
makespan objective is fundamentally incompatible with incentives in
prior-free settings.
In this context, our work
can be viewed as giving a meaningful approach for obtaining positive 
results that are close to prior-free for a problem for which most
results are very negative.

Given these strong negative results, several special cases of the
problem have been studied. Lavi and Swamy~\cite{LS07} give constant
factor approximations when job sizes can take on only two different
values. Lu and Yu~\cite{LY08, LY08-2, Lu09} consider the problem over
two machines, and give approximation ratios strictly better than $2$.

%
%
Related machine scheduling is the special case where the runtime of a
job on a machine is the product of the machine's private speed and the
job's public length.  Importantly, the private information of each
machine in a related machine scheduling problem is single-dimensional,
and the total length of the jobs assigned to any given machine in the
makespan minimizing schedule is monotone in the machine's speed.  This
monotonicity implies that the related machine makespan objective is
incentive compatible (i.e., the price of anarchy is one).  For this
reason work on related machine scheduling has focused on computational
tractability.  Archer and Tardos \cite{AT01} give a constant
approximation mechanism and Dhangwotnotai et al.~\cite{DDDR08} give an
incentive compatible polynomial time approximation scheme thereby
matching the best approximation result absent incentives.  There are
no known approximation-preserving black-box reductions from mechanism
design to algorithm design for related machine scheduling; moreover,
in the Bayesian model Chawla, Immorlica, and Lucier~\cite{CIL12}
recently showed that the makespan objective does not admit black-box
reductions of the form that Hartline and Lucier \cite{HL10} showed
exist for the objective of social welfare maximization.


Another line of work studies the makespan objective subject to an
envy-freedom constraint instead of the incentive-compatibility
constraint.  A schedule and payments (to the machines) are envy free if
every machine prefers its own assignment and payment to that of
any other machine.  Mu'alem \cite{M09} introduced the envy-free
scheduling problem for makespan.  Cohen et al.~\cite{CFFKO10} gave a
polynomial time algorithm for computing an envy-free schedule that is
an $O(\log m)$ approximation to the first-best makespan (i.e., the
optimal makespan absent envy-freedom constraints).  Fiat and
Levavi~\cite{FL12} complement this by showing that the optimal
envy-free makespan (a.k.a. second-best makespan) can be an
$\Omega(\log m)$ factor larger than the first-best makespan.

\section{Preliminaries and main results}

\label{sec:prelim}

We consider the scheduling of $n$ jobs on $m$ unrelated machines where
the running time of a job on a machine is drawn from a distribution.
A {\em schedule} is an assignment of each job to exactly one machine.
The {\em load} of a machine is the number of jobs assigned to it.  The
{\em load factor} is the average number of jobs per machine and
is denoted $\load = n/m$.  The {\em work} of a machine is the sum of the
runtimes of jobs assigned to it.  The {\em total work} is the sum of
the works of each machine.  The {\em makespan} is the most work assigned
to any machine.

The vector of running times for each of the jobs on a given machine is
that machine's private information.  A scheduling mechanism may
solicit this information from the machines, may make payments to the
machines, and must select a schedule of jobs on the machines.  A
scheduling mechanism is evaluated in the equilibrium of strategic
behavior of the machines.  A particularly robust equilibrium concept
is dominant strategy equilibrium.  A scheduling mechanism is {\em
  incentive compatible} if it is a dominant strategy for each machine
to report its true processing time for each job.

We consider the following simple mechanisms:
\begin{description}
\item[minimum work]
 The minimum work mechanism solicits the running times, selects the
 schedule to minimize the total work and pays each machine its
 externality, i.e., the difference between the minimum total work when
 the machine does nothing and the total work of all other machines in
 the selected schedule.
\item[bounded overload] The bounded overload mechanism is
  parameterized by an overload factor $\overload > 1$ and is
  identical to the minimum work mechanism except it optimizes subject
  to placing at most $\overload\load$ jobs on any machine.
\item[sieve / anonymous reserve] 
The sieve mechanism, also known as the anonymous reserve mechanism, is
parameterized by a reserve $\reserve \geq 0$ and is identical to the
minimum work mechanism except that there is a dummy machine added with
runtime $\reserve$ for all jobs.  Jobs assigned to the dummy machine
are considered unscheduled.
\item[sieve and bounded overload] The sieve and bounded overload
  mechanism is parameterized by overload $\overload$, reserve
  $\reserve$, and a partition parameter $\partparam$.  It partitions
  the machines into two sets of sizes $(1-\partparam) m$ and
  $\partparam m$.  It runs the sieve with reserve $\reserve$ on the
  first set of machines and runs the bounded overload mechanism with
  overload $\overload$ on the unscheduled jobs and the second set of
  machines.
\end{description}
The above mechanisms are incentive compatible.  The minimum work
mechanism is incentive compatible as it is a special case of the well
known Vickrey-Clarke-Groves (VCG) mechanism which is incentive
compatible.  The bounded overload mechanism is what is known as a
``maximal in range'' mechanism and is also incentive compatible (by
the VCG argument).  The sieve / anonymous reserve mechanism is
incentive compatible because the incentives of the agents in the
minimum work mechanism are unaffected by the addition of a dummy
agent.  Finally, the sieve and bounded overload mechanism is incentive
compatible because from each machine's perspective it is either
participating in the sieve mechanism or the bounded overload
mechanism.

The runtimes of jobs on machines are drawn from a product distribution
(a.k.a., the {\em prior}) that is symmetric with respect to the
machines.  (Therefore, the running times of a job on each machine are
i.i.d.\@ random variables.)  The distribution of job $j$ on any
machine is denoted $\distj$; a draw from this distribution is denoted
$\Tj$.  The {\em best runtime} of a job is its minimum runtime over
all machines, this first order statistic of $m$ random draws from
$\distj$ is denoted by $\Tji[1]$.

Our goal is to exhibit a mechanism that is prior-independent and a
good approximation to the expected makespan of the best incentive
compatible mechanism for the prior, i.e., the second-best solution.
Because both the second-best and the first-best expected makespans are
difficult to analyze, we will give our approximation via one of the
following two lower bounds on the first-best solution.
\begin{description}
\item[expected worst best runtime]  
The expected worst best runtime is the expected value of the best runtime of
the job with the longest best runtime, i.e., $\expect[\max_j \Tji[1]]$
\item[expected average best runtime]
The expected average best runtime is the expected value of the sum of the
best runtimes of each job averaged over all machines, i.e.,
$\expect[\sum_j \Tji[1]] / m$.
\end{description}
Intuitively, the former gives a good bound when the load factor is
small, the latter when the load factor is large.  We will refer to any
of these bounds on the first-best makespan as OPT, with the
assumption that which of the bounds is meant, if it is important, is
clear from the context.

As an intermediary in our analysis of the makespan of our scheduling
mechanisms with respect to OPT, we will give bicriteria results that
compare our mechanism's makespan to the makespan of an optimal
schedule with fewer machines.  This restriction is well defined
because the machines are a prior identical.  For a given parameter
$\partparam$, $\OPTr$ will denote the optimal schedule with
$\partparam m$ machines (via bounds as described above).  Much of our
analysis will be with respect to $\OPTh$, i.e., the optimal schedule
with half the number of machines.

While it is possible to construct distributions where $\OPT$ is much
smaller than $\OPTh$, for many common distributions they are quite
close.  In fact, for the class of distributions that satisfy the
monotone hazard rate (MHR) condition,\footnote{The hazard rate of a
distribution $F$ is given by $h(x)=\frac{f(x)}{1-F(x)}$, where $f$ is
the probability density function for $F$; a distribution $F$ satisfies
the MHR condition if $h(x)$ is non-decreasing in $x$. Many natural
distributions such as the uniform, Gaussian, and exponential
distributions, satisfy the monotone hazard rate
condition. Intuitively, these are distributions with tails no heavier
than the exponential distribution.} $\OPT$ and $\OPTh$ are always
within a factor of four; more generally $\OPT$ and $\OPTr$ are within
a factor of $1/\partparam^2$ for these distributions. (See proof in
Section~\ref{sec:probabilistic}.)

\begin{lemma}
\label{lem:MHR}
  When the distributions of job sizes have monotone hazard rates the
  expected worst best and average best runtimes on $\partparam m$ machines are
  no more than $1/\partparam^2$ times the expected worst best and average best
  runtimes, respectively, on $m$ machines. 
\end{lemma}

\subsection{Main Results}
Our main theorems are as follows.  When the number of jobs is
comparable to the number of machines, i.e., the load factor $\load$ is
constant, then the bounded overload mechanism is a good approximation
to the optimal makespan on $m/2$ machines.

\begin{theorem}
\label{thm:main-noniid} 
For $n$ jobs, $m$ machines, load factor $\load = n/m$, and runtimes
distributed according to a machine-symmetric product distribution, the
expected makespan of the bounded overload mechanism with overload
$\overload = \overloadvalue$ is a $\finalApprox\load$ approximation to the expected worst
best runtime, and hence also to the optimal makespan, on $m/2$
machines. 
\end{theorem}

\begin{corollary}
  Under the assumptions of Theorem~\ref{thm:main-noniid} where
  additionally the distributions of job sizes have monotone hazard
  rates, the expected makespan of the bounded overload mechanism with
  $\overload = \overloadvalue$ is a $\finalApproxMHR\load$ approximation
  to the expected optimal
  makespan. 
\end{corollary}

When the load factor $\load$ is large and the job runtimes are
identically distributed, the sieve and bounded overload mechanism is a
good approximation to the optimal makespan.  The following theorems
and corollaries demonstrate the sieve and bounded overload mechanism
under two relevant parameter settings.

\begin{theorem}\label{thm:replication}
  For $n$ jobs, $m$ machines, and runtimes from an i.i.d.\@
  distribution, the expected makespan of the sieve and bounded
  overload mechanism with overload $\overload = \overloadvalue$, partition
  parameter $\partparam=2/3$, and reserve~$\reserve = \frac{n}{m\log m}
  \expect[\dmin{\frac{\partparam}{2} m}{\T}]$ is an $O(\sqrt {\log
    m})$ approximation to the larger of the expected worst best
  and average best runtime, and hence also to the optimal
  makespan, on $m/3$ machines. Here $\T$ denotes a draw from the
  distribution on job sizes. 
\end{theorem}

\begin{corollary}
  Under the assumptions of Theorem~\ref{thm:replication} where
  additionally the distribution of job sizes has monotone hazard rate,
  the expected makespan of the sieve and bounded overload mechanism is
  an $O(\sqrt{\log m})$ approximation to the expected optimal makespan.
\end{corollary}

\begin{theorem}\label{thm:llm}
  For $n \geq m \log m$ jobs, $m$ machines, and runtimes from an
  i.i.d.\@ distribution, the expected makespan of the sieve and
  bounded overload mechanism with overload $\overload = \overloadvalue$, partition
  parameter $\partparam= 1/\log \log m$, and reserve\\
  $\reserve=\frac{2n}{m\log m} \expect[\dmin {\frac\partparam{2}
    m}\T]$, is a constant approximation to the larger of the expected
  worst best and average best runtime, and hence also to
  the optimal makespan, on $\partparam m/2$ machines. Here $\T$
  denotes a draw from the distribution on job sizes.
\end{theorem}

\begin{corollary} 
  Under the assumptions of Theorem~\ref{thm:llm} where additionally
  the distribution of job sizes has monotone hazard rate the expected
  makespan of the sieve and bounded overload mechanism is a $O((\log
  \log m)^2)$ approximation to the expected optimal makespan.
\end{corollary}

We prove Theorem~\ref{thm:main-noniid} in Section~\ref{sec:small-m}
and Theorems~\ref{thm:replication} and \ref{thm:llm} in
Section~\ref{sec:general}.

\subsection{Probabilistic Analysis}
\label{sec:prob}
Our goal is to show that the simple processes described by the bounded
overload and sieve mechanisms result in good makespan and our
upper bound on makespan is given by the first order statistics of each
job's runtime across the machines.  The sieve's performance analysis
is additionally governed by the law of large numbers.  We describe
here basic facts about order statistics and concentration bounds.
Additionally we give a number of new bounds, proofs of which are in
Section~\ref{sec:probabilistic}.

For random variable $X$ and integer $k$, we consider the following
basic constructions of $k$ independent draws of the random variable.
The $i$th order statistic, or the $i$th minimum of $k$ draws, is
denoted $X\ithofn{i}{k}$.  The first order statistic, i.e., the
minimum of the $k$ draws, is denoted $\dmin{k} X$.  The $k$th order
statistic, i.e., the maximum of $k$ draws, is denoted $\dmax{k} X$.
Finally, the sum of $k$ draws is denoted $\dsum{k} X$.  We include the
possibility that $i$ or $k$ can be random variables.  We also allow
the notation to cascade, e.g., for the special case where the jobs are
i.i.d. from $\dist$ the lower bounds on $\OPT$ are
$\dmax{n}{\dmin{m}{\T}}$ and $\dsum{n}{\dmin{m}{\T}}/m$ for the
expected worst best and average best runtime, respectively, and $\T$
drawn from $\dist$.



We will use the following forms of Chernoff-Hoeffding bounds in this paper.
Let $X = \sum_i X_i$, where $X_i \in [0,B]$ are independent random
variables. Then, for all $\epsilon \geq 1$, 
\begin{equation*}
\Pr[X > (1+\epsilon)\expect[X]] <
\exp\left(\tfrac{-\epsilon\expect[X]}{3B}\right) < \exp\left(\tfrac{-(1+\epsilon)\expect[X]}{6B}\right)
\end{equation*}




Our analysis often involves relating different order statistics of a
random variable (e.g. how does the size of a job on its best machine
compare to that on its second best machine). We relate these different
order statistics via the stochastic dominance relation. This is useful
in our analysis because stochastic dominance is preserved by the max
and sum operators. We say that a random variable $X$ is stochastically
dominated by another random variable $Y$ if for all $t$, $\pr[X\le
t]\ge \pr[Y\le t]$. Stochastic dominance is equivalent to being able
to couple the two random variables $X$ and $Y$ so that $X$ is always
smaller than $Y$.

Below, the first lemma relates the $i$th order statistic over some
number of draws to the first order statistic over half the draws.  The
second relates the minimum over several draws of a random variable
to a single draw of that variable.  The third 
relates the
maximum over multiple draws of a random variable to an appropriate sum over those
draws.  These lemmas are proved in Section~\ref{sec:probabilistic}.


\begin{lemma}
\label{lem:stoch-dom}
Let $X$ be any nonnegative random variable and $m$ and $i\le m$ be
arbitrary integers. Let $\alpha$ be defined such that
$\Pr[X\le\alpha]=1/m$ (or for discontinuous distributions, $\alpha = \sup\{z:\Pr[X \le z] < 1/m\}$). Then
$X\ithofn{i}{m}$ is stochastically dominated by
$\max(\alpha,\dmax{\couplingApprox^i}{\dmin{m/2} X})$.
\end{lemma}

\begin{lemma}\label{thm:MHR}
 For a random variable $X$ whose distribution satisfies the monotone
 hazard rate condition, $X$ is stochastically dominated by $r
 \dmin{r} X$. 
\end{lemma}

\begin{lemma}
\label{lem:noniidAdjustment}
Let $K_1,\cdots,K_n$ be independent and identically
distributed integer random variables such that for some constant
$c>1$, we have $K_j \geq c$, and let
$\rvw_1,\cdots,\rvw_n$ be arbitrary independent nonnegative variables. Then,
\begin{align*}
\expect\left[\max\nolimits_j \dmax{K_j}{\rvw_j}\right]
&\leq
\tfrac{c}{c-1}\expect\left[K_1\right]\expect\left[\max\nolimits_j \rvw_j\right].
\end{align*}
\end{lemma}

%

We will analyze the expected makespan of a mechanism as the maximum
over a number of correlated real-valued random variables. The
correlation among these variables makes it difficult to understand and
bound the makespan. Our approach will be to replace these random
variables with an ensemble of independent random variables that have
the same marginal distributions. Fortunately, this operation does not
change the expected maximum by too much. Our next lemma relates the
expected maximum over an arbitrary set of random variables to the
expected maximum over a set of independent variables with the same
marginal distributions. It is a simple extension of the correlation
gap results of Aggarwal et al.~\cite{ADSY10}, Yan\cite{Y11}, and Chawla et
al.~\cite{CHMS10}.


\begin{lemma}
\label{lem:corGapVar}
Let $X_1,\cdots,X_n$ be arbitrary correlated real-valued random
variables. Let $Y_1,\cdots,Y_n$ be independent random variables
defined so that the distribution of $Y_i$ is identical to that of
$X_i$ for all $i$.   Then, $\expect[\max_j X_j]\le \frac{e}{e-1}\expect[\max_j Y_j]$.
\end{lemma}

\section{The bounded overload mechanism}
\label{sec:small-m}

\newcommand{\step}{s}
\newcommand{\placed}[1][\step]{P_{#1}}
\newcommand{\displaced}[1][\step]{D_{#1}}
\newcommand{\nSteps}{N}
\newcommand{\nHeads}{H}

Recall that the bounded overload mechanism minimizes the total work subject to
the additional feasibility constraint that every machine is assigned at most
$\overload\load$ jobs. In this section we prove that the expected makespan of
the bounded overload mechanism, with the overload set to
$\overload=\overloadvalue$, is a
$\finalApprox\load$ factor approximation to the expected best worst runtime and thus to
the optimal makespan.


Intuitively the bounded overload mechanism tries to achieve two
objectives simultaneously: (1) keep the size of every job on the
machine its schedule to be close to its size on its best machine, but
also (2) evenly distribute the jobs across all the machines.  Recall,
that the minimum work mechanism achieves the first objective exactly,
but fails on the second objective.  Due to the independence between
jobs, the number of jobs on each machine may be quite unevenly
distributed.  In contrast, the bounded overload mechanism explicitly
disallows uneven assignments of jobs and therefore the main issue
to address in its analysis is whether it satisfies the first
objective, i.e., that the sizes of the jobs are close to what they are
in the minimum work mechanism.


To setup for the proof of
Theorem~\ref{thm:main-noniid} consider the following definitions that
describe the outcome of the bounded overload mechanism and the worst
best runtime on $m/2$ machines (which bounds the optimal makespan on
$m/2$ machines).  Let $\Tj$ denote a random variable drawn according
to job $j$'s distribution of runtimes $\distj$.  Let $\Bj$ denote the
job's best runtime out of $m/2$ machines, i.e., $\Bj =
\Tj\ithofn{1}{m/2}$, the first order statistic of $m/2$ draws.  The
expected worst best runtime on $m/2$ machines is
$\expect[\max_j \Bj]$.  The bounded overload mechanism considers
placing each job on one of $m$ machines.  These runtimes of job $j$
drawn i.i.d.\@ from $\distj$ impose a (uniformly random) ordering over
the machines starting from the machine that is ``best'' for $j$ to the
one that is ``worst''; this is $j$'s preference list. Let
$\Tji[\rank]$ denote the size of job $j$ on the $\rank$th machine in
this ordering (also called the job's $\rank$th favorite machine).  Let
$\Rj$ be a random variable to denote the rank of the machine that job
$j$ is placed on by the bounded overload mechanism.  As each machine
is constrained to receive at most $\overload\load$ jobs, the expected makespan of
bounded overload is $\overload\load\expect[\max_j\Tji[\Rj]]$. We will
bound this quantity in terms of $\expect[\max_j \Bj]$.

There are three main parts to our argument. First, we note that the
$\Rj$s are correlated across different $j$'s, and so are the
$\Tji[\Rj]$s. This makes it challenging to directly analyze
$\expect[\max_j\Tji[\Rj]]$. We use Lemma~\ref{lem:corGapVar} to
replace the $\Rj$s in this expression by independent random variables
with the same marginal distributions. We then show that the marginal
distributions can be bounded by simple geometric random variables
$\Rhatj$. To do so, we introduce another procedure for assigning jobs
to machines that we call the last entry procedure. The assignment of
each job under the last entry procedure is no better than its
assignment under bounded overload. On the other hand, the ranks of the
machines to which jobs are allocated in the last entry procedure are
geometric random variables with a bounded failure rate. Finally, we
relate the runtimes $\Tji[\Rhatj]$ to the optimal runtimes $\Bj$ using
Lemma~\ref{lem:stoch-dom}.

We begin by describing the last entry procedure.



\begin{description}
\item[last entry] In order to schedule job $j$, we first apply the
  bounded overload mechanism $\VCGc$ to all jobs other than $j$. We
  then place $j$ on the first machine in its preference list that has
  fewer than $\overload\load$ jobs. Let $\Lj$ denote the rank of the
  machine to which $j$ gets allocated.
\end{description}
We now make a few observations about the ranks $\Lj$ realized by the
last entry procedure.
\begin{lemma}
\label{lem:lastEntry}
The runtime of any job $j$ in bounded overload is no worse than its runtime in
the last entry procedure. That is, $\Rj\le\Lj$.
\end{lemma}
\begin{proof}
  Fix any instantiation of jobs' runtimes over machines. Consider the
  assignment of job $j$ in the last entry procedure, and let LE($j$)
  denote the schedule where all of the jobs but $j$ are scheduled
  according to bounded overload and $j$ is scheduled according to the
  last entry procedure. Since the bounded overload mechanism minimizes
  total work, the total runtime of all of the jobs in $\VCGc$ is no
  more than the total runtime of all of the jobs in LE($j$). On the
  other hand, the total runtime of all jobs except $j$ in LE($j$) is
  no more than the total runtime of all jobs except $j$ in
  $\VCGc$. This immediately implies that $j$'s runtime in bounded
  overload is no more than its runtime in last entry. Since this holds
  for any fixed instantiation of runtimes, we have $\Rj\le\Lj$.
\end{proof}


Next, we show that the rank $\Lj$ of a job $j$ in last entry is
stochastically dominated by a geometric random variable $\Rhatj$ that
is capped at $\lceil\frac{m}{\overload}\rceil$. Note that $\Lj$ is at
most $\lceil\frac{m}{\overload}\rceil$ since
$\lceil\frac{m}{\overload}\rceil$ machines can accommodate $\lceil
\frac{m}{c}\rceil c\load \geq n$ jobs and therefore last entry will
never have to send a job to anything worse than its $\lceil
\frac{m}{c}\rceil$th favorite machine.  The random variable $\Rhatj$
also lives in $\{1,\dots,\lceil \frac{m}{c}\rceil\}$, and is drawn
independently for all $j$ as follows: for $i\in\{1,\ldots,\lceil
\frac{m}{c}\rceil-1\}$, we have $\pr[\Rhatj=i] =
\frac{1-1/c}{c^{i-1}}$; and the remaining probability mass is on
$\lceil \frac{m}{c} \rceil$.
\begin{lemma}
\label{lem:lastEntry-ub}
The rank $\Lj$ of a job $j$ in last entry is stochastically dominated
by $\Rhatj$, and so the runtime of job $j$ in last entry is
stochastically dominated by $\Tji[\Rhatj]$.
\end{lemma}
\begin{proof}
  We use the principle of deferred decisions. In order to schedule
  $j$, the last entry procedure first runs bounded overload on all of
  the jobs other than $j$. This produces a schedule in which at most a
  $\frac{1}{\overload}$ fraction of the machines have all of their
  slots occupied. Conditioned on this schedule, job $j$'s preference
  list over machines is a uniformly random permutation. So the
  probability (over the draw of $j$'s runtimes) that job $j$'s
  favorite machine is fully occupied is at most
  $1/\overload$. Likewise, the probability that the job's two most
  favorite machines are both occupied is at most $1/\overload^2$, and
  so on. Therefore, the rank of the machine on which $j$ is eventually
  scheduled is dominated by a geometric random variable with failure
  rate $1/\overload$.
\end{proof}
Lemmas~\ref{lem:lastEntry} and~\ref{lem:lastEntry-ub} yield the following
corollary. 
\begin{corollary}
\label{cor:vcg-ub}
For all $j$, the runtime $\Tji[\Rj]$ of job $j$ in bounded overload is stochastically dominated by $\Tji[\Rhatj]$.
\end{corollary}

The benefit of relating $\Tji[\Rj]$s with $\Tji[\Rhatj]$s is that
while the former are correlated random variables, the latter are
independent, because the $\Rhatj$'s are picked
independently. Corollary~\ref{cor:vcg-ub} implies that we can replace
the former with the latter, gaining independence, while losing only a
constant factor in expected makespan.


\begin{corollary}
\label{lem:corGap}
$\expect[\max_j \Tji[\Rj]]$ is no more than $e/(e-1)$ times $\expect[\max_j \Tji[\Rhatj]]$.
\end{corollary}

The final part of our analysis relates the $\Tji[\Rhatj]$s to the
$\Bj$s. A natural inequality to aim for is to bound
$\ex[\Tji[\Rhatj]]$ from above by a constant times $\ex[\Bj]$ for each
$j$.  Unfortunately, this is not enough for our purposes: note that
our goal is to upper bound $\ex[\max_j \Tji[\Rhatj]]$ in terms of
$\ex[\max_j \Bj]$. Thus we proceed to show that $\Tji[\Rhatj]$ is
stochastically dominated by a maximum among some number of copies of
$\Bj$. We apply Lemma~\ref{lem:stoch-dom} (stated in
Section~\ref{sec:prelim} and proved in
Section~\ref{sec:probabilistic}) to the random variable $\Tji$ for
this purpose. Define $\tj=\sup\{t:\distj(t) < 1/m\}$.
Then the lemma shows that $\Tji$ is stochastically dominated by \linebreak
$\max(\tj,\dmax{\couplingApprox^i}{\Bj})$. 

Let $D_j$ be defined as
$\couplingApprox^{\Rhatj}$. Note that $\expect[D_j]$ can be bounded by
a constant whenever
$\overload > 4$ (this upper bound is obtained by treating $\Rhatj$ as a
geometric random variable without being capped at $\lceil\frac{m}{\overload}\rceil$).  Then
Lemma~\ref{lem:stoch-dom} implies the following corollary.
\begin{lemma}
\label{lem:stoch-dom-general}
$\Tji[\Rhatj]$ is stochastically dominated by $\max(\tj,\dmax{D_j} \Bj)$.
\end{lemma}

\noindent
We are now ready to prove the main theorem of this section. 

 \begin{oneshot}{Theorem~\ref{thm:main-noniid}}
 For $n$ jobs, $m$ machines, load factor $\load = n/m$, and runtimes
 distributed according to a machine-symmetric product distribution, the
 expected makespan of the bounded overload mechanism with overload
 $\overload = \overloadvalue$ is a $\finalApprox\load$ approximation to the expected worst
 best runtime, and hence also to the optimal makespan, on $m/2$
 machines. 
 \end{oneshot}

\begin{proof}
The proof follows from the following series of inequalities that we explain below.
First we have $\msp(\VCGc) \leq \overload\load \expect\left[\max_j
\Tji[\Rj]\right]$ by the fact that $\VCGc$ schedules at most $\overload\load$
jobs per machine
\begin{align*}
\frac{e-1}{e}\expect\left[\max_j \Tji[\Rj]\right]
&\le \expect\left[\max_j \Tji[\Rhatj]\right]\\
& \le \expect\left[\max_j (\max(\tj,\dmax{D_j} \Bj))\right]\\
&\le \expect\left[\max_j \left(\tj+\dmax{D_j} \Bj\right)\right]\\
&\leq \max_j \tj + \expect\left[\max_j \dmax{D_j} \Bj\right]\\
&\leq 2\OPTh  + \tfrac{\couplingApprox}{\couplingApprox-2}\expect[D_j]\expect[\max_j\Bj]\\
&\leq \left(2+8\frac{\overload-1}{\overload-4}\right)\OPTh.
\end{align*}
The first of the inequalities follows from Lemma~\ref{lem:corGap},
the second from Lemma~\ref{lem:stoch-dom-general}, the third from
noting that the maximum of non-negative random variables is upper
bounded by their sum, and the last by the definition of $\OPTh$, along with the 
fact that $\expect[D_j] \leq 4\frac{\overload-1}{\overload-4}$. For
the fifth inequality we use Lemma~\ref{lem:noniidAdjustment} to bound the second term. 
For the first term in
that inequality consider the job $j$ that has the largest $\tj$. For
this job, the probability that its size on all of the $m/2$ machines
in \OPTh\ is at least $\tj$ is $(1-\distj(\tj))^{m/2} \geq
(1-1/m)^{m/2}\ge 1/2$ by the definition of $\tj$. So
$\OPTh\ge \max_j \tj/2$. 

The final approximation factor therefore is
$\overload\load\frac{e}{e-1}\left(2+8\frac{\overload-1}{\overload-4}\right)$ for all $\overload > 4$.
At $\overload = \overloadvalue$, this evaluates to a factor $\finalApprox\load$ approximation.
\end{proof}

\section{The sieve and bounded overload mechanism}
\label{sec:general}

We will now analyze the performance of the sieve and bounded overload
mechanisms under the assumption that the jobs are a priori
identical. Let us consider the sieve mechanism first. Recall that this
is essentially the minimum work mechanism where every job is assigned
to its best machine, except that jobs with a size larger than $\beta$
on every machine are left unscheduled. The bound of $\beta$ on the
size of scheduled jobs allows us to employ concentration results to
bound the expected makespan of the mechanism. Changing the value of
$\beta$ allows us to tradeoff the makespan of the mechanism with the
number of unscheduled jobs.

\begin{lemma}
  \label{lem:sieve} For $k<\log m$, the expected makespan of the sieve
  mechanism with $\beta=\frac{n\expect[\dmin m\T]}{km}$ is no more
  than $O(\log m/k)$ times the expected average best runtime, and
  hence also the expected optimal makespan. The expected number of jobs
  left unscheduled by the mechanism is $km$.
\end{lemma}
\begin{proof}
  Let us first consider the expected total work of any single machine,
  that is the expected total size of jobs scheduled on that
  machine. Let $Y_{ij}$ be a random variable that takes on the value
  $0$ if job $j$ is not scheduled on machine $i$, and takes on the
  size of $j$ on machine $i$ if the job is scheduled on that
  machine. The probability that $j$ is scheduled on $i$ is no more
  than $1/m$; its expected size on $i$ conditioned on being scheduled
  is at most $\tau=\expect[\dmin m\T$]. Therefore, $\expect[\sum_j
  Y_{ij}]\le \frac{n\tau}{m}$, which in turn is at most the average
  best runtime.

  Note that the $Y_{ij}$'s are independent and bounded random
  variables. So we can apply Chernoff-Hoeffding bounds and use $\beta
  = \frac{n\tau}{km}$ to get
\begin{align*}
\Pr\left[\sum\nolimits_j Y_{ij} >  \tfrac{7\log m}{k}\OPT \right]
\le 
\Pr\left[\sum\nolimits_j Y_{ij} >  \tfrac{7\log m}{k}\expect\left[\sum\nolimits_j Y_{ij}\right]
\right] 
< \exp\left( -\tfrac 13 \tfrac{6\log m}{k}
  \tfrac{n\tau}{\beta m}  \right) = \tfrac{1}{m^2}.
\end{align*}
Taking the union bound over the $m$ machines, we get that with
probability $1-1/m$, the makespan of the sieve mechanism is at most
$O(\log m/k)$ times $\OPT$. 

We will now convert this tail probability into a bound on the expected
makespan. Let $\gamma$ denote the factor by which the expected
makespan of the mechanism 
exceeds $\OPT$. Remove all jobs with best runtimes greater than
$\beta$ from consideration and consider creating sieve's schedule by
assigning each of the leftover jobs to their best machine (minimizing
total work) one-by-one in decreasing order of best runtime, until the
makespan exceeds $\frac{7}{k} \log m$ times $\OPT$. This event happens with a
probability at most $1/m$. When this event happens, we are left with a
smaller set of jobs; conditioned on being left over at this point,
these jobs have a smaller best runtime than the average over all
scheduled jobs. Thus the expected makespan for scheduling them will be
at most $\gamma\OPT$.  So we get $\gamma \leq 7\log m/k + \gamma/m$,
i.e., $\gamma = O(\log m/k)$. This implies the first part of the
lemma.

We now prove the second part of the lemma, i.e., the expected number of jobs
left unscheduled is $km$. Note that $\beta$ exceeds a
job's expected best runtime by a factor of $n/km$. Thus by applying Markov's
inequality, we get the probability of a job's best runtime being larger than
$\beta$ to be at most $km/n$. Hence the expected number of jobs with best
runtime larger than $\beta$ is $km$. 
\end{proof}

Next we will combine the sieve mechanism with the bounded overload
mechanism. We consider two different choices of parameters. Note that
if in expectation the sieve mechanism leaves $km$ jobs unscheduled,
using the bounded overload mechanism to schedule these jobs over a set
of $\Omega(m)$ machines gives us an expected makespan that is at most
$O(k)$ larger than the expected optimal makespan on that number of
machines. In order to balance this with the makespan achieved by
sieve, we pick $k=\sqrt{\log m}$. This gives us
Theorem~\ref{thm:replication}.

 \begin{oneshot}{Theorem~\ref{thm:replication}}
  For $n$ jobs, $m$ machines, and runtimes from an i.i.d.\@
   distribution, the expected makespan of the sieve and bounded
   overload mechanism with overload $\overload = \overloadvalue$, partition
   parameter $\partparam=2/3$, and reserve $\reserve=\frac{n}{m\log m}
   \expect[\dmin{\frac{\partparam}{2} m}{\T}]$, is an $O(\sqrt {\log m})$
   approximation to the larger of the worst best runtime and the
   average best runtime, and hence also to the optimal makespan, on
   $m/3$ machines. Here $\T$ denotes a draw from the distribution on
   job sizes.
 \end{oneshot}

\begin{proof}
  For the choice of parameters in the theorem statement, we use $m/3$
  of the $m$ machines for the sieve mechanism, and the remainder for
  the bounded overload mechanism. The expected makespan of the overall
  mechanism is no more than the sum of the expected makespans of the
  two constituent mechanisms. Lemma~\ref{lem:sieve} implies that the
  expected makespan of the sieve mechanism is \linebreak $O(\sqrt{\log
    m})$ times $\OPT_{1/3}$, and the load factor for the bounded overload mechanism
  is also $O(\sqrt{\log m})$. Theorem~\ref{thm:main-noniid} then
  implies that the expected makespan of the bounded overload mechanism
  is also $O(\sqrt{\log m})$ times $\OPT_{1/3}$.
\end{proof}

If we partition the machines across the sieve and the bounded overload
mechanisms roughly equally, then Theorem~\ref{thm:replication} gives us the optimal
choice for the parameter $\beta$. A different possibility is to
perform a more aggressive screening of jobs by using a smaller
$\beta$, while comparing our performance against a more heavily
penalized optimal mechanism -- one that is allowed to use only a
$\partparam$ fraction of the machines. 

 \begin{oneshot}{Theorem~\ref{thm:llm}}
   For $n \geq m \log m$ jobs, $m$ machines, and runtimes from an
   i.i.d.\@ distribution, the expected makespan of the sieve and
   bounded overload mechanism with overload $\overload = \overloadvalue$, partition
   parameter $\partparam= 1/\log \log m$, and reserve
   $\reserve=\frac{2n}{m\log m} \expect[\dmin {\frac\partparam{2}
     m}\T]$, is a constant approximation to the larger of the worst
   best runtime and the average best runtime, and hence also to the
   optimal makespan, on $\partparam m/2$ machines. Here $\T$ denotes a
   draw from the distribution on job sizes.
 \end{oneshot}
\begin{proof}
  We will show that the expected makespan of the sieve mechanism is at most a
  constant times the average best runtime on $\partparam m/2$
  machines, and the expected number of unscheduled jobs is
  $O(\partparam m)$. The current theorem then follows by applying
  Theorem~\ref{thm:main-noniid}.

  Let us analyze the expected makespan of the sieve mechanism first. Let $\tau=
  \expect[\dmin {\frac{\partparam}{2} m}\T]$. Then we can bound $\OPT_{\partparam/2}$ as
  $\OPT_{\partparam/2}\ge\frac{2n\tau}{\partparam m}$. As in the proof of
  Lemma~\ref{lem:sieve}, let $Y_{ij}$ be a random variable that takes on the
  value $0$ if job $j$ is not scheduled on machine $i$, and takes on
  the size of $j$ on machine $i$ if the job is scheduled on that
  machine. Then, 
\begin{align*}
\expect\left[\sum\nolimits_j Y_{ij}\right] & \le
  \tfrac{n}{(1-\partparam)m}\expect[\dmin {(1-\partparam) m}\T]
\le
  \tfrac{n\tau}{(1-\partparam)m}\le \tfrac{\partparam}{2(1-\partparam)}
  \OPT_{\partparam/2}.
\end{align*}
\noindent
  Applying Chernoff-Hoeffding bounds we get
\begin{align*}
\Pr\left[\sum\nolimits_j Y_{ij} > 2\OPT_{\partparam/2} \right] 
\le 
\Pr\left[\sum\nolimits_j Y_{ij} > 4(1/\partparam - 1)\expect\left[\sum\nolimits_j Y_{ij}\right]
\right] 
< \exp\left( - \tfrac 1{\partparam}
  \tfrac{n\tau}{(1-\partparam) m}\tfrac{1}{\beta}  \right) \le m^{-1/2\partparam}.
\end{align*}
Here we used $\beta=2n\tau/m\log m$.  Taking the union bound over the $m$
machines, we get that with probability $o(1)$, the makespan of the
sieve mechanism is at most twice $\OPT_{\partparam/2}$. Once again, as
in the proof of Lemma~\ref{lem:sieve} we can convert this tail bound into a
constant factor bound on the expected makespan.

Now let us consider the jobs left unscheduled. For any given job, we
will compute the probability that its runtime on all of the
$(1-\partparam)m$ machines is larger than $\beta$. Because $\beta$ is
defined in terms of $\dmin {\frac{\partparam}{2} m}\T$, we will consider the
machines in batches of size $\partparam m/2$ at a time. Using Markov's
inequality, the probability that the job's runtime exceeds $\beta$ on
all machines in a single batch is at most $m\log m/2n$. There are
$2(1/\partparam -1)$ batches in all, so the probability that a job
remains unscheduled is at most $(m\log m/n)(2^{2(1-1/\partparam)})$, which by our
choice of $\partparam$ is $O(\partparam m/n)$.
\end{proof}


\section{Deferred proofs}
\label{sec:probabilistic}
In this section we prove the bounds for random variables and order statistics from Section~\ref{sec:prob}.

 \begin{oneshot}{Lemma~\ref{lem:stoch-dom}}
 Let $X$ be any nonnegative random variable, and $m$,\ $i\le m$ be
 arbitrary integers. Let $\alpha$ be defined such that
 $\Pr[X\le\alpha]=1/m$ (or for discontinuous distributions, 
 $\alpha = \sup\{z:\Pr[X \le z] < 1/m\}$). Then
 $X\ithofn{i}{m}$ is stochastically dominated by
 $\max(\alpha,\dmax{\couplingApprox^i}{\dmin{m/2} X})$.
 \end{oneshot}
\begin{proof}
Let $\dist$ be the cumulative distribution function of $X$. 
We prove this by showing that $X\ithofn{i}{m}$ is ``almost''
stochastically dominated by $\dmax{\couplingApprox^i}{\dmin{m/2} X}$; specifically, we
show that for all $t \geq \alpha$,
\begin{equation*}
\Pr\left[X\ithofn{i}{m} > t\right] \leq \Pr\left[\dmax{\couplingApprox^i}{\dmin{m/2} X} > t\right].
\end{equation*}
To prove this inequality, we will define a process for instantiating
the variables $X\ithofn{i}{m}$ and $\dmax{\couplingApprox^i}{\dmin{m/2} X}$ in a
correlated fashion such that the former is always larger than the
other.

$\dmax{\couplingApprox^i}{\dmin{m/2} X}$ is a statistic based on $\couplingApprox^i m/2$
independent draws of the random variable $X$. Consider partitioning
these draws into $\couplingApprox^i/2$ groups of size $m$ each. We then randomly
split each group into two smaller groups, which we will refer to as
blocks, of size $m/2$ each. Define a good event $\G$ to be the event
that at least one of these $\couplingApprox^i/2$ groups get split such that the $i$
smallest runtimes in it all fall into the same block. If event $\G$
occurs, arbitrarily choose one group which caused event $\G$, and for
all $k$ define $X\ithofn{k}{m}$ to be the $k$th min from this group.
Otherwise, select an arbitrary group to define the $X\ithofn{k}{m}$.
Note that since we split the groups into blocks randomly, and this is
independent of the drawn runtimes in the groups, $X\ithofn{k}{m}$ has
the correct distribution, both when $\G$ occurs and does not occur.
Define the minimum from each of the $\couplingApprox^i$ blocks to be a draw of
$\dmin{m/2} X$.  Thus, whenever $\G$ occurs, the probability that the
$\dmax{\couplingApprox^i}{\dmin{m/2} X} > t$ is at least the probability that
$X\ithofn{i+1}{m} > t$. We have that
\begin{align*}
\Pr\left[\dmax{\couplingApprox^i}{\dmin{m/2} X} > t \right]
&\geq \Pr\left[\G\right]\cdot \Pr\left[X\ithofn{i+1}{m} > t\right]\\
&= \left(\Pr\left[\G\right]\cdot \frac{\Pr\left[X\ithofn{i+1}{m} >
t\right]}{\Pr\left[X\ithofn{i}{m} > t\right]}\right)\cdot\Pr\left[X\ithofn{i}{m}
> t\right].
\end{align*}
We now show that $\left(\Pr\left[\G\right]\cdot \frac{\Pr\left[X\ithofn{i+1}{m}
>t\right]}{\Pr\left[X\ithofn{i}{m} > t\right]}\right) \geq 1$ whenever $\dist(t) \geq 1/m$, which completes our proof of the lemma. Note that 
\begin{align*}
\frac{\Pr\left[X\ithofn{i+1}{m} > t\right]}{\Pr\left[X\ithofn{i}{m} > t\right]}
&= \frac{\sum_{k=0}^i \binom{m}{k}\dist(t)^k(1-\dist(t))^{m-k}}{\sum_{k=0}^{i-1}
\binom{m}{k}\dist(t)^k(1-\dist(t))^{m-k}}\\
&= 1 + \frac{\binom{m}{i}\dist(t)^i(1-\dist(t))^{m-i}}{\sum_{k=0}^{i-1}
\binom{m}{k}\dist(t)^k(1-\dist(t))^{m-k}},
\end{align*}
which we can see is an increasing function of $\dist(t)$. Thus in the 
range $\dist(t) \geq 1/m$, it
attains its minimum precisely at
$\dist(t) = 1/m$. Substituting $F(t) = 1/m$ into the above, and using
standard approximations for $\binom{m}{k}$ (namely $\left(\frac{m}{k}\right)^k
\leq \binom{m}{k} \leq \left(\frac{me}{k}\right)^k$, we have
\begin{align*}
\frac{\Pr\left[X\ithofn{i+1}{m} > t\right]}{\Pr\left[X\ithofn{i}{m} > t\right]}
&\geq 1 +
\frac{\left(\frac{m}{i}\right)^i\left(\frac{1}{m}\right)^i\left(1-\frac{1}{m}\right)^{m-i}}
{\left(1-\frac{1}{m}\right)^m +
{\displaystyle \sum_{k=1}^{i-1}}\left(\frac{me}{k}\right)^k\left(\frac{1}{m}\right)^k\left(1-\frac{1}{m}\right)^{m-k}}\\
&\geq 1 + \frac{i^{-i}}{1 + (i-1)\cdot\max_k(\frac{e}{k})^k}
\geq 1 + \frac{i^{-i}}{1 + (i-1)e}.
\end{align*}
It suffices to show that this last quantity, when multiplied with
$\Pr[\G]$, is at least $1$.  We consider the complement of event $\G$,
call it even $\B$. The event $\B$ occurs only when none of the $\couplingApprox^i/2$
groups split favorably.  The probability that a group splits favorably (for $i \geq 1$)
is 
$\left.
{2\cdot\binom{m-i}{m/2-i}}
\middle/
{\binom{m}{m/2}}
\right. \geq
2^{-(i-1)}$. So we can see that  $\Pr[\B] \leq (1-2^{-(i-1)})^{\couplingApprox^i/2} \leq
e^{-(\couplingApprox/2)^i}$, and thus $\Pr[\G] \geq 1 - e^{-(\couplingApprox/2)^i}$. 
It can be verified that $(1-e^{-(\couplingApprox/2)^i})\cdot\left(1 + \frac{i^{-i}}{1 +
(i-1)e}\right) \geq 1$. 
\end{proof}

\begin{oneshot}{Lemma~\ref{thm:MHR}}
For a random variable $X$ whose distribution satisfies the monotone hazard rate condition, $X$ is stochastically dominated by $r\dmin{r} X$.
\end{oneshot}       
\begin{proof}
The hazard rate function is related to the cumulative distribution
function as $\Pr[X\geq t] = e^{-\int_0^t h(z) \,dz}$.
Likewise, we can write:
\begin{align*}
\Pr[r\dmin{r}X \geq t]  = 
\Pr[\dmin{r}X \geq t/r] 
 = (e^{-\int_0^{t/r} h(z) \,dz})^r =
e^{-r\int_0^{t/r} h(z) \,dz}.
\end{align*}
In order to prove the lemma, we need only show that 
$\int_0^t h(z) \,dz \geq r\cdot\int_0^{t/r} h(z)\,dz$. Since the hazard rate
function $h(z)$ is monotone, the function $\int_0^t h(z)\,dz$ is a convex
function of $t$. The required inequality follows from the definition of convexity.
\end{proof}

 \begin{oneshot}{Lemma~\ref{lem:noniidAdjustment}}
   Let $K_1,\cdots,K_n$ be independent and identically
   distributed integer random variables such that for some constant
   $c>1$, we have $K_j \geq c$ for all $j$, and let
   $\rvw_1,\cdots,\rvw_n$ be arbitrary independent nonnegative
   variables. Then,
 \begin{align*}
 \expect\left[\max\nolimits_j \dmax{K_j}{\rvw_j}\right]
 &\leq
 \tfrac{c}{c-1}\expect\left[K_1\right]\expect\left[\max\nolimits_j \rvw_j\right].
 \end{align*}
 \end{oneshot}
\begin{proof}
 We consider the following process for generating correlated
 samples for $\max_j \rvw_j$ and $\max_j \dmax{K_j} {\rvw_j}$. We first
 independently instantiate $K_j$ for every $j$; recall that these are
 identically distributed variables. Let $k=\sum_j K_j\ge cn$. Then we
 consider all possible $n!$ permutations of these instantiated values. For
 each permutation $\sigma$, we make the corresponding number of
 independent draws of the random variable $\rvw_j$ for all $j$; call
 this set of draws $X_{\sigma}$. In all, we get $k n!$ draws from
 the distributions, that is, $|\cup_{\sigma} X_{\sigma}| = k
 n!$. Exactly $k (n-1)!$ of these draws belong to any particular
 $j$; denote these by $Y_j$.

 Now, the maximum element out of each of the $X_{\sigma}$ sets is an
 independent draw from the same distribution $\max_j
 \dmax{K_j}{\rvw_j}$ is drawn from. We get $n!$ independent samples from that
 distribution. Call this set of samples $X$.

 Next note that each set $Y_j$ contains $k (n-1)!$ independent draws
 from the distribution corresponding to $\rvw_j$. We construct a
 uniformly random $n$-dimensional matching over the sets $Y_j$, and
 from each $n$-tuple in this matching we pick the maximum. Each such
 maximum is an independent draw from the distribution corresponding
 to $\max_j\rvw_j$, and we get $k (n-1)!$ such samples; call this set
 of samples $Y$.

 Finally, we claim that $\expect[\sum_{y\in Y} y] \ge
 (1-1/c)\expect[\sum_{x\in X} x]$, with the expectation taken over the
 randomness in generating the $n$-dimensional matching across the
 $Y_j$s. The lemma follows, since we have $\expect[\sum_{x\in X} x] =
 n!\expect[\max_j \dmax{K_j}{\rvw_j}]$ as well as
\begin{align*}
\expect[\sum_{y\in Y} y] 
&= \expect_{\{K_j\}}[k(n-1)! \expect[\max_j\rvw_j]] 
= n! \expect[K_j]\expect[\max_j\rvw_j].
\end{align*}
 To prove the claim, we call an $x\in X$ ``good'' if the $n$-tuple in
 the matching over $\{Y_j\}$ that it belongs to does not contain any
 other element of $X$. Then, $\expect[\sum_{y\in Y} y] \ge
 \expect[\sum_{x\in X} x\Pr[x\text{ is ``good"}]]$. 

 Let us compute the probability that some $x$ is ``good''. Without
 loss of generality, suppose that $x\in Y_1$. In order for $x$ to be
 good, it's $n$-tuple must not contain any of the other elements of
 $X$ from the other $Y_j$'s. If we define $x_j=|X\cap Y_j|$, then
 $\Pr[x\text{ is ``good"}]$ is at least $\prod_{j\ne 1}
 (1-\frac{x_j}{k(n-1)!})$ where $\sum x_j\le n!$. This product is
 minimized when we set one of the $x_j$s to $n!$ and the rest to $0$,
 and takes on a minimum value of $1-n/k\ge 1-1/c$.
\end{proof}

\begin{oneshot}{Lemma~\ref{lem:corGapVar}}
  Let $X_1,\cdots,X_n$ be arbitrary correlated real-valued random
  variables. Let $Y_1,\cdots,Y_n$ be independent random variables
  defined so that the distribution of $Y_i$ is identical to that of
  $X_i$ for all $i$.  Then, $\expect[\max_j X_j]\le
  \frac{e}{e-1}\expect[\max_j Y_j]$.
\end{oneshot}
\begin{proof}
  We use the following result from \cite{ADSY10} (also implicit in
  \cite{CHMS10}). Let $U$ be a universe of $n$ elements, $f$ a
  monotone increasing submodular function over subsets of this
  universe, and $D$ a distribution over subsets of $U$. Let
  $\tilde{D}$ be a product distribution (that is, every element is
  picked independently to draw a set from this distribution) such that
  $\Pr_{S\sim D}[i\in S]=\Pr_{S\sim \tilde{D}}[i\in S]$. Then
  $\expect_{S\sim D}[f(S)]\le \frac{e}{e-1}\expect_{S\sim
    \tilde{D}}[f(S)]$.

  To apply this theorem, let us first assume that the variables $X_i$
  are discrete random variables over a finite domain. The universe $U$
  will then have one element for each possible instantiation of each
  variable $X_i$ with a value equal to that instantiation. Then any
  joint instantiation of the variables $X_1, \cdots, X_n$ corresponds
  to a subset of $U$; let $D$ denote the corresponding distribution
  over subsets. Let $f$ be the max function over the instantiated
  subset. Then $\expect[\max_j X_j]$ is exactly equal to
  $\expect_{S\sim D}[f(S)]$. As before, let $\tilde{D}$ denote the
  distribution over subsets of $U$ where each element is picked
  independently. Likewise, the random variables $Y_1, \cdots, Y_n$
  define a distribution, say $D'$, over subsets of $U$. Note that
  under $D'$ the memberships of elements of $U$ in the instantiated
  subset are negatively correlated -- for two elements that correspond
  to instantiations of the same variable, including one in the subset
  implies that the other is not included. This raises the expected
  maximum. In other words, $\expect_{S\sim D'}[f(S)]\ge \expect_{S\sim
    \tilde{D}}[f(S)]$. Therefore, we get $\expect[\max_j
  X_j]=\expect_{S\sim D}[f(S)]\le (e/e-1) \expect_{S\sim
    D'}[f(S)]=(e/e-1)\expect[\max_j Y_j]$.

  When the variables $X_j$ are defined over a continuous but bounded
  domain, we can apply the above argument to an arbitrarily fine
  discretization of the variables. Our claim then follows from taking
  the limit as the granularity of the discretization goes to zero.

  Finally, let us address the boundedness assumption. For some
  $\epsilon<1/n^2$, let $B$ be defined so that for all $i$, $\Pr[X_i>
  B]\le \epsilon$. Then the contribution to the expected maximum from
  values above $B$ is similar for the $X$s and the $Y$s: the
  probability that some variable $X_i$ attains the maximum value $b>B$
  is at most $\Pr[X_i=b]$ whereas the probability that the variable
  $Y_i$ attains the maximum value $b>B$ is at least
  $(1-\epsilon)^{n-1}\Pr[Y_i=b]$. Therefore, $\expect[\max_j
  X_j]\le (1+o(\epsilon))(e/e-1)\expect[\max_j Y_j]$. Taking the limit
  as $\epsilon$ goes to zero implies the theorem.
\end{proof}

\paragraph{Comparing $\OPT$ and $\OPTr$}
We now prove Lemma~\ref{lem:MHR}.  The key intuition behind the lemma is that
it can be viewed as the result of scaling both sides of the stochastic
dominance relation of Lemma~\ref{thm:MHR} up by a constant, and as we
shall see, the monotone hazard rate condition is retained by the
minimum among multiple draws from a probability distribution.
 \begin{oneshot}{Lemma~\ref{lem:MHR}} 
   When the distributions of job sizes have monotone hazard rates the
   expected worst best and average best runtimes on $\partparam m$ machines are
   no more than $1/\partparam^2$ times the expected worst best and average best
   runtimes respectively on $m$ machines. 
 \end{oneshot}       

\begin{proof}
  We will show that the random variable $\dmin{\partparam m}\Tj$ is
  stochastically dominated by $\frac 1{\partparam}\dmin{m}\Tj$. Then,
  the expected worst best runtime with $\partparam m$ machines is no
  more than $1/\delta$ times the expected worst best runtime with $m$
  machines. Likewise, the expected average best runtime with
  $\partparam m$ machines is no more than $1/\delta^2$ times the
  expected average best runtime with $m$ machines. (The extra
  $1/\delta$ factor comes about because we average over $\delta m$
  machines for the former, versus over $m$ machines for for the
  latter.)

Our desired stochastic dominance relation is precisely of the form
given by Lemma~\ref{thm:MHR}.  In particular, observe that taking a
minimum among $m$ draws is exactly the same as first splitting the $m$
draws into $1/\partparam$ groups, selecting the minimum from each
group of $\partparam m$ draws, and then taking the minimum from this
collection of $1/\partparam$ values.  Thus, we can see that
$(1/\partparam)\dmin{m}\Tj=(1/\partparam)\dmin{1/\partparam}{\dmin{\partparam
m}\Tj}$, and so the claim follows immediately from Lemma~\ref{thm:MHR}
as long as the distribution of $\dmin{\partparam m}\Tj$ has a monotone
hazard rate.  We show in Claim~\ref{cl:minMHR} below that the first order statistic of
i.i.d. monotone hazard rate distributions also has a monotone hazard
rate.
\end{proof} 


\begin{claim}\label{cl:minMHR}
A distribution $F$ has a monotone hazard rate if and only if the distribution of the 
minimum among $k$ draws from $F$ has a monotone hazard rate.
\end{claim}
\begin{proof}
 Let $F_k$ denote the cdf for minimum among $n$ draws from $F$. Then
 we have $F_k(x) = 1 - (1-F(x))^k$, and the corresponding $f_k(x) =
 k(1-F(x))^{k-1}f(x)$. Thus the hazard rate function is:
 $$h_k(x) = \frac{f_k(x)}{1-F_k(x)} = \frac{k(1-F(x))^{k-1}f(x)}{(1-F(x))^k} =
 k\frac{f(x)}{1-F(x)}.$$
 This is precisely $k$ times the hazard rate function $h(x)$, and
 therefore, $h_k(x)$ is monotone increasing if and only if $h(x)$ is.
\end{proof}

\section{Conclusions}
Non-linear objectives coupled with multi-dimensional preferences
present a significant challenge in mechanism design. Our work shows
that this challenge can be overcome for the makespan objective when
agents (machines) are a priori identical.  This suggests a number of
interesting directions for follow-up. Is the gap between the
first-best and second-best solutions (i.e.~the cost of incentive
compatibility) still small when agents are not identical? Does
knowledge of the prior help?  Note that this question is meaningful
even if we ignore computational efficiency. On the other hand, even if
the gap is small, the optimal incentive compatible mechanism may be
too complex to find or implement. In that case, can we approximate the
optimal incentive compatible mechanism in polynomial time?

Similar questions can be asked for other non-linear objectives. One
particularly interesting objective is max-min fairness, or in the
context of scheduling, maximizing the running time of the least loaded
machine. Unlike for makespan, in this case we cannot simply
``discard'' a machine (that is, schedule no jobs on it) without
hurting the objective. This necessitates techniques different from the
ones developed in this paper.

\bibliographystyle{plain}
\bibliography{msp}


\end{document}